\documentclass[letterpaper, 10pt, conference]{ieeeconf}
\IEEEoverridecommandlockouts
\usepackage{etex}

\usepackage{graphicx}
\usepackage{balance}
\usepackage{cite}
\usepackage{soul}
\usepackage{epstopdf}
\usepackage{amssymb}
\usepackage[cmex10]{amsmath}
\usepackage{cases}
\let\labelindent\relax
\usepackage{enumitem}
\usepackage{booktabs}
\usepackage{algorithm}
\usepackage{algpseudocode}

\usepackage{array}
\usepackage{multirow}
\usepackage{bigdelim}
\usepackage{mdwtab}
\usepackage{tikz}
\usepackage{physics}
\usepackage{color}

\makeatletter
\let\MYcaption\@makecaption
\makeatother
\usepackage[font=footnotesize]{subcaption}
\makeatletter
\let\@makecaption\MYcaption
\makeatother

\usepackage{bm}

\DeclareMathOperator*{\argmax}{arg\,max}

\DeclareMathOperator*{\argsup}{arg\,sup}

\usepackage{accents}

\newtheorem{theorem}{Theorem}
\newtheorem{lemma}{Lemma}

\newtheorem{assumption}{Assumption}

\newtheorem{remark}{Remark}

\title{\bfseries Smoothing-Averse Control: Covertness and Privacy from Smoothers}%
\author{Timothy L.~Molloy and Girish N.~Nair%
\thanks{The authors are with the Department of Electrical and Electronic Engineering, University of Melbourne, Parkville, VIC, 3010, Australia.
{\texttt{\{tim.molloy,gnair\}@unimelb.edu.au}}}
\thanks{This work received funding from the Australian Government, via grant AUSMURIB000001 associated with ONR MURI grant N00014-19-1-2571.}%
}

\begin{document}




\maketitle
\thispagestyle{empty}
\pagestyle{empty}

\begin{abstract}

\boldmath
In this paper we investigate the problem of controlling a partially observed stochastic dynamical system such that its state is difficult to infer using a (fixed-interval) Bayesian smoother.
This problem arises naturally in applications in which it is desirable to keep the entire state trajectory of a system concealed.
We pose our smoothing-averse control problem as the problem of maximising the (joint) entropy of smoother state estimates (i.e., the joint conditional entropy of the state trajectory given the history of measurements and controls).
We show that the entropy of Bayesian smoother estimates for general nonlinear state-space models can be expressed as the sum of entropies of marginal state estimates given by Bayesian filters.
This novel additive form allows us to reformulate the smoothing-averse control problem as a fully observed stochastic optimal control problem in terms of the usual concept of the information (or belief) state, and solve the resulting problem via dynamic programming.
We illustrate the applicability of smoothing-averse control to privacy in cloud-based control and covert robotic navigation.
\end{abstract}

\newcommand\copyrighttext{%
  \footnotesize \textcopyright 2021 IEEE. Personal use of this material is permitted.
  Permission from IEEE must be obtained for all other uses, in any current or future
  media, including reprinting/republishing this material for advertising or promotional
  purposes, creating new collective works, for resale or redistribution to servers or
  lists, or reuse of any copyrighted component of this work in other works.}
\newcommand\copyrightnotice{%
\begin{tikzpicture}[remember picture,overlay]
\node[anchor=south,yshift=5pt] at (current page.south) {\fbox{\parbox{\dimexpr\textwidth-\fboxsep-\fboxrule\relax}{\copyrighttext}}};
\end{tikzpicture}%
}

%
\IEEEpeerreviewmaketitle
\copyrightnotice

\section{Introduction}

Controlling or limiting the information that cyber-physical systems disclose during their operation is a challenging yet increasingly important problem across many application domains including robotics and autonomous vehicles \cite{Karabag2019,Li2019,Savas2020,Hibbard2019,Krishnamurthy2019,Mattila2020,Lourenco2020}, networked control \cite{Mochaourab2018,Tanaka2017}, smart grid and power systems \cite{Li2018,Poor2017,Farokhi2018}, and sensor networks \cite{Leong2019}.
Concealing information about the state of a cyber-physical system with dynamics (i.e.~a dynamical system) is particularly challenging since system inputs and/or outputs from isolated time instances have the potential to reveal information about the entire state trajectory through correlations introduced by the system dynamics.
The design of both controllers \cite{Tanaka2017,Tanaka2018,Hibbard2019,Karabag2019,Li2019,Savas2020} and output filters \cite{Mochaourab2018,Leong2019} that limit the disclosure of dynamical system state information through inputs and/or outputs has therefore attracted considerable recent attention (see also \cite{Nekouei2019,Farokhi2020} and references therein).
Despite these efforts, few works have addressed the problem of how best to control a system to conceal its entire state trajectory from an adversary that employs a Bayesian smoother for state trajectory estimation.
We investigate this smoothing-averse control problem in this paper.

Recent work on limiting the disclosure of state information in dynamical systems has primarily been motivated by privacy considerations.
Privacy concepts in dynamical systems have been developed by drawing on ideas of privacy for static settings (e.g.~datasets) and include differential privacy \cite{Sandberg2015,Farokhi2020,Hale2015} and information-theoretic privacy \cite{Nekouei2019,Farokhi2020,Tanaka2017,Murguia2021}.
Information-theoretic privacy approaches based on concepts such as mutual information \cite{Murguia2021}, directed information \cite{Nekouei2019,Tanaka2017}, and Fisher information \cite{Farokhi2018} have proved particularly attractive for dynamical systems since they lead to provable privacy guarantees \cite{Farokhi2020}.
However, many of these approaches, particularly those based on directed information (cf.~\cite{Nekouei2019,Tanaka2017}), do not explicitly consider the potential for system inputs and/or outputs from isolated time instances to reveal information about the entire state trajectory.
These approaches are thus vulnerable to leaking information to Bayesian smoothers which can infer states from past, present, and future measurements and controls (cf.~\cite{Briers2010}).

The problem of covert motion planning for robots and other agents has been investigated in parallel to privacy in dynamical systems.
Whilst many works on covert motion planning have been developed in the context of robot navigation with deterministic descriptions of adversary sensing capabilities (e.g., sensor field of view) \cite{Marzouqi2006,Marzouqi2011,Rost2019}, covert motion planning is increasingly being considered with approaches inspired by information-theoretic privacy.
For example, information-theoretic approaches that seek to make the future trajectory or policy of an agent difficult to predict and/or infer have been proposed in \cite{Savas2018,Savas2020,Hibbard2019} on the basis of maximising the control policy entropy, and in \cite{Karabag2019} on the basis of Fisher information.
One of the few approaches that explicitly seeks to conceal an agent's entire trajectory is proposed in \cite{Li2019} and involves minimising the total probability mass (over a specific time horizon) that the state estimates at individual times from a Bayesian filter have at the true states.
This total probability mass criterion leaves open the possibility for adversaries to accurately infer the state at some isolated time instances and then use knowledge of the state dynamics to improve their estimates of the entire state trajectory (e.g., via Bayesian smoothing).

The problem of controlling a system so as to keep its entire state trajectory private or concealed from a Bayesian smoother is thus not fully resolved but is of considerable practical importance.
Indeed, it is frequently desirable to conceal both the instantaneous state of a system along with the states it has previously visited and the states it will visit.
Furthermore, Bayesian smoothers are widely used by state-of-the-art target tracking systems (cf.~\cite{Bar-Shalom2001}) and robot localisation systems (cf.~\cite{Thrun2005}), making them one of the inference approaches most likely to be employed by an adversary.

The main contribution of this paper is the proposal of a new information-theoretic smoothing-averse control problem and the novel reformulation of it as a fully observed stochastic optimal control problem in terms of the standard concept of the information (or belief) state (i.e., the marginal conditional distribution of the partially observed state given the history of measurements and controls, as given by a Bayesian filter).
We specifically consider the problem of maximising the entropy of the posterior distribution over state trajectories computed by a Bayesian smoother since a large value of this smoother entropy is naturally suggestive of a high degree of uncertainty in the trajectory estimates.
Maximising this smoother entropy within a stochastic optimal control problem via dynamic programming is however non-trivial with it not admitting an obvious additive expression in terms of the information state.
Indeed, dynamic programming usually relies on costs (or rewards) being added at each time step and depending only on the current information state and control (cf.\ \cite{Bertsekas2005}).
The structure of the costs required to facilitate efficient dynamic programming has resulted in limited past success synthesising practical control policies that deliver privacy or covertness outside of linear systems (cf.~\cite{Tanaka2017,Li2018,Farokhi2018}) or systems with discrete finite state spaces (cf.~\cite{Savas2018,Savas2020,Hibbard2019,Karabag2019,Li2019}).
Key to our reformulation of our smoothing-averse control problem as a fully observed stochastic optimal control solvable via dynamic programming is thus the development of a novel additive expression for the (joint) entropy of Bayesian smoother estimates for general nonlinear state-space models (which is potentially of independent interest).
We illustrate the applicability of our smoothing-averse control problem to privacy in cloud-based control (as in \cite{Tanaka2017}) and to covert robotic navigation.

This paper is structured as follows:
in Section \ref{sec:problem}, we pose our smoothing-averse control problem; in Section \ref{sec:active}; we reformulate our smoothing-averse control problem as an optimal control solvable via dynamic programming; finally, in Section \ref{sec:examples}, we illustrate the application of smoothing-averse control before presenting conclusions in Section \ref{sec:conclusion}.

\subsection*{Notation}
We shall denote random variables in this paper with capital letters such as $X$, and their realisations with lower case letters such as $x$.
Sequences of random variables will be denoted by capital letters with superscripts to denote their final index (e.g., $X^T \triangleq \{X_0, X_1, \ldots, X_T\}$), and lowercase letters with superscripts will denote sequences of realisations (e.g., $x^T \triangleq \{x_0, x_1, \ldots, x_T\}$).
All random variables will be assumed to have probability densities (or probability mass functions in the discrete case).
The probability density (or mass function) of a random variable $X$ will be written $p(x)$, the joint probability density of $X$ and $Y$ written as $p(x, y)$, and the conditional probability density of $X$ given $Y = y$ written as $p(x|y)$ or $p(x | Y = y)$.
If $f$ is a function of $X$ then the expectation of $f$ evaluated with $p(x)$ will be denoted $E_X [f(x)]$ and the conditional expectation of $f$ evaluated with $p(x|y)$ as $E[f(x) | y]$.
The {\em point-wise} (differential) entropy of a continuous random variable $X$ given $Y = y$ will be written $h(X | y) \triangleq - \int p(x|y) \log p(x|y)\, dx$ with the (average) conditional entropy of $X$ given $Y$ being $h(X | Y) \triangleq E_{Y} \left[ h(X|y) \right]$. The entropy of discrete random variables is defined similarly with summation instead of integration.
The mutual information between $X$ and $Y$ is $I(X; Y) \triangleq h(X) - h(X | Y) = h(Y) - h(Y | X)$ and the \emph{point-wise} conditional mutual information of $X$ and $Y$ given $Z = z$ is $I(X; Y | z) \triangleq h(X | z) - h(X | Y, z)$ with the (average) conditional mutual information given by $I(X; Y | Z) \triangleq E_{Z} \left[ I(X; Y | z) \right]$.

\section{Problem Formulation}
\label{sec:problem}
Let us consider the discrete-time stochastic system
\begin{subequations}
\label{eq:sys}
\begin{align}
 X_{t+1}
 &= f_t \left( X_t, U_t, W_t \right)\\
 Y_t
 &= g_t(X_t, U_{t-1}, V_t)
\end{align}
\end{subequations}
for $t \geq 0$ where $f_t : \mathbb{R}^n \times \mathbb{R}^m \times \mathbb{R}^q \mapsto \mathbb{R}^n$ and $g_t : \mathbb{R}^n \times \mathbb{R}^m \times \mathbb{R}^r \mapsto \mathbb{R}^N$ are (possibly nonlinear) functions, $X_t \in \mathcal{X} \subset \mathbb{R}^n$ are state vectors, $U_t \in \mathcal{U} \subset \mathbb{R}^m$ are control variables, and $Y_t \in \mathcal{Y} \subset \mathbb{R}^N$ are measurements.
We assume that $W_t \in \mathbb{R}^q$ and $V_t \in \mathbb{R}^r$ are each independent but not necessarily identically distributed stochastic noise processes (that are mutually independent).
The state process $X_t$ thus forms a first-order (generalised) Markov chain with initial state density (or probability mass function if $\mathcal{X}$ is discrete) $X_0 \sim p(x_0)$ and state transition kernel
\begin{align}
    \label{eq:stateProcess}
    X_{t+1} | (X_t = x_t, U_t = u_t) \sim p( x_{t+1} | x_t, u_t).
\end{align}
The observations $Y_t$ are also conditionally independent given the state $X_t = x_t$ with conditional probability density
\begin{align}
    \label{eq:obsProcess}
    Y_t | (X_t = x_t, U_{t-1} = u_{t-1}) \sim p( y_t | x_t, u_{t-1}).
\end{align}

In addition to the system \eqref{eq:sys}, consider an adversary that seeks to estimate the state of the system $X_t$ given measurements from a second observation process
\begin{align}
    \label{eq:advProcess}
    Z_t | (X_t = x_t, U_{t-1} = u_{t-1}) \sim p(z_t | x_t, u_{t-1}).
\end{align}
The adversary seeks to estimate the state trajectory $X^T$ of the system over some (known) finite horizon $T$ by employing a standard (fixed-interval) Bayesian smoother that computes the posterior density $p(x^T | z^T)$ (or distribution when $\mathcal{X}$ is discrete) over state trajectory realisations $x^T$ given measurement trajectory realisations $z^T$.
In this paper, we shall consider a novel smoothing-averse control problem in which the system \eqref{eq:sys} is controlled to increase the uncertainty in the estimates of the state trajectory $X^T$ given $Z^T$.

To quantify the uncertainty associated with estimates of the state trajectory, we shall use the conditional entropy
$
    h(X^T | Z^T) 
    = E_{Z^T}[h(X^T | z^T)]
$
where $h(X^T | z^T)$ is the entropy of the smoother density (or distribution) $p(x^T | z^T)$.
We will assume that the specific sensing process of the adversary is unknown, but that it results in greater uncertainty about the state trajectory $X^T$ than the system's own measurements and controls $Y^T$ and $U^{T-1}$.
Our assumption about the sensing capability of the adversary is summarised as follows.
\begin{assumption}
\label{assumption:measurements}
The adversary has access to measurements $Z_t$ from a sensing process that is less informative about the state trajectory $X^T$ than the system's own measurements and controls $Y_t$ and $U_t$ in the sense that
\begin{align*}
    h(X^T | Z^T) \geq h(X^T | Y^T, U^{T-1}).
\end{align*}
\end{assumption}

In light of Assumption \ref{assumption:measurements}, we shall specifically consider a worst-case adversary that is omniscient with direct access to the system's own measurements $Y_t$ and controls $U_t$ without any extra noise (which contrasts with differential privacy approaches that perturb signals observed by adversaries).
For this adversary, the entropy $h(X^T | Z^T)$ will be equivalent to the smoother entropy $h(X^T | Y^T, U^{T-1})$.
To obfuscate the state trajectory, we seek to maximise the smoother entropy $h(X^T | Y^T, U^{T-1})$ using only the control inputs (in contrast to adding noise to the signals observed by the adversary as in differential privacy approaches).
Our smoothing-averse control problem is therefore to find a (potentially stochastic) output feedback control policy $\mu \triangleq \left\{ \mu_t : 0 \leq t < T \right\}$ defined by the conditional probability densities $\mu_t(y^t, u^{t-1}) \triangleq p(u_t | y^t, u^{t-1})$ that solves
\begin{align}
 \label{eq:smoothingAverse}
 \begin{aligned}
   \sup \left\{ h(X^T | Y^T, U^{T-1}) - \gamma E \left[ \sum_{t = 0}^{T-1} c_t \left(x_t, u_t\right) \right] \right\}
 \end{aligned}
\end{align}
subject to the dynamic constraints \eqref{eq:sys} and the state and control constraints $X_t \in \mathcal{X}$ and $U_t \in \mathcal{U}$.
Here, the expectation $E[\cdot]$ is over the states and controls $X^{T-1}$ and $U^{T-1}$, and $c_t : \mathcal{X} \times \mathcal{U} \mapsto [0, \infty)$ for $0 \leq t \leq T - 1$ are instantaneous cost functions (weighted by some parameter $\gamma \geq 0$) that can encode objectives other than increasing the smoother entropy (e.g., keeping the state close to a desired trajectory).

The key novelty of \eqref{eq:smoothingAverse} is our consideration of the smoother entropy $h(X^T | Y^T, U^{T-1})$ as a measure of state-trajectory uncertainty, and its optimisation over the controls $U^{T-1}$ to obfuscate the entire state trajectory.
A key difficulty with solving \eqref{eq:smoothingAverse} is that the smoother entropy $h(X^T | Y^T, U^{T-1})$ is not in the form of a sum of costs (or rewards) added over time nor is it a pure terminal cost.
A direct application of dynamic programming is therefore not possible and naive suboptimal numerical solutions for solving \eqref{eq:smoothingAverse} require the implementation of Bayesian smoothers to compute $p(x^T | y^T, u^{T-1})$ and approximate $h(X^T | Y^{T}, U^{T-1})$.
To overcome these difficulties, we shall develop a new solution approach by establishing a novel additive form for $h(X^T | Y^T, U^{T-1})$ involving conditional probabilities (i.e., information states) calculable with Bayesian filters.
We specifically avoid the naive reduction of $h(X^T | Y^T, U^{T-1})$ to the sum of $h(X_t | Y^t, U^{t-1})$ for $0 \leq t \leq T$ since this reduction only holds under extremely restrictive independence assumptions.


\section{Stochastic Optimal Control\\ Reformulation and Solution}
\label{sec:active}

In this section, we cast our smoothing-averse control problem \eqref{eq:smoothingAverse} into the framework of stochastic optimal control and examine its solution using dynamic programming.

\subsection{Smoother Entropy Reformulation}

To reformulate our smoothing-averse control problem, let us first establish a novel additive form of the smoother entropy $h(X^T | Y^{T-1}, U^{T-1})$ for the state-space system \eqref{eq:sys}.

\begin{lemma}
\label{lemma:stageAdditive}
 The smoother entropy $h(X^T | Y^T, U^{T-1})$ for the state-space system \eqref{eq:sys} has the additive form:
 \begin{align}\notag
        &h(X^T | Y^T, U^{T-1})\\\notag
        &\quad= E_{Y^T, U^{T-1}} \Bigg[ \sum_{t = 0}^T \left[ h(X_t | y^{t}, u^{t-1}) - h(X_t | y^{t-1}, u^{t-1}) \right.\\\label{eq:stageAdditive}
        &\qquad\qquad\qquad\qquad\left. + h(X_t | X_{t-1}, y^{t-1}, u^{t-1}) \right]\Bigg]
\end{align}
where we have that $h(X_0 | Y^{0}, U^{-1}) = h(X_0 | Y_0)$ and $h(X_0 | Y^{-1}, U^{-1}) = h(X_0 | X_{-1}, Y^{-1}, U^{-1}) = h(X_0)$.
\end{lemma}
\begin{proof}
Due to the expectation relationship between conditional entropies and the entropies of conditional distributions, it suffices to show that
 \begin{align}\notag
        &h(X^T | Y^T, U^{T-1})\\\notag
        &\quad= \sum_{t = 0}^T \left[ h(X_t | Y^{t}, U^{t-1}) - h(X_t | Y^{t-1}, U^{t-1}) \right.\\\label{eq:stageAdditiveExp}
        &\qquad\qquad\qquad\qquad\left. + h(X_t | X_{t-1}, Y^{t-1}, U^{t-1}) \right].
\end{align}
We shall therefore prove \eqref{eq:stageAdditiveExp} via induction on $T$.
We first note that for $T = 0$,
   \begin{align*}
       &h(X_0 | Y^{0}, U^{-1}) - h(X_0 | Y^{-1}, U^{-1})\\ 
       &\quad+ h(X_0 | X_{-1}, Y^{-1}, U^{-1})\\
       &\qquad= h(X_0 | Y_0) - h(X_0) + h(X_0)
  \end{align*}
and so \eqref{eq:stageAdditiveExp} holds when $T = 0$.
Suppose then that \eqref{eq:stageAdditiveExp} holds for sequence lengths smaller than $T$ where $T \geq 1$.
Then,
 \begin{align}\notag
        &\sum_{t = 0}^T \left[ h(X_t | Y^{t}, U^{t-1}) - h(X_t | Y^{t-1}, U^{t-1}) \right.\\\notag
        &\quad\left. + h(X_t | X_{t-1}, Y^{t-1}, U^{t-1}) \right]\\\notag
        &= h(X^{T-1} | Y^{T-1}, U^{T-2}) + h(X_T | Y^T, U^{T-1})\\\notag
        &\quad- h(X_T | Y^{T-1}, U^{T-1}) + h(X_T | X_{T-1}, Y^{T-1}, U^{T-1})\\\notag
        &= h(X^{T-1} | Y^{T-1}, U^{T-1}) + h(X_T | Y^T, U^{T-1})\\\label{eq:tempCondEnt}
        &\quad- h(X_T | Y^{T-1}, U^{T-1}) + h(X_T | X^{T-1}, Y^{T-1}, U^{T-1})
\end{align}
where the first equality holds due to the induction hypothesis and the second holds since: 1)
$
    h(X^{T-1} | Y^{T-1}, U^{T-1})
    = h(X^{T-1} | Y^{T-1}, U^{T-2})
$
due to $U_{T-1}$ being conditionally independent of $X^{T-1}$ given $Y^{T-1}$ and $U^{T-2}$ by virtue of the permitted control policy; and, 2) the Markov property of the state process implies that
$
    h(X_T | X_{T-1}, Y^{T-1}, U^{T-1})
    = h(X_T | X^{T-1}, Y^{T-1}, U^{T-1})
    (= h(X_T | X_{T-1}, U_{T-1})).
$
Recalling the definition of mutual information, we have that
\begin{align*}
    &h(X_T | Y^{T-1}, U^{T-1}) - h(X_T | X^{T-1}, Y^{T-1}, U^{T-1})\\
    &\;= I(X_T; X^{T-1} | Y^{T-1}, U^{T-1})\\
    &\;= h(X^{T-1} | Y^{T-1}, U^{T-1}) - h(X^{T-1} | X_T, Y^{T-1}, U^{T-1}) \\
    &\;= h(X^{T-1} | Y^{T-1}, U^{T-1}) - h(X^{T-1} | X_T, Y^T, U^{T-1})
\end{align*}
where the last line holds due to the measurement kernel \eqref{eq:obsProcess} which implies that the measurement $Y_T$ is conditionally independent of $X^{T-1}$ given $X_T$, $U^{T-1}$, and $Y^{T-1}$.
Substituting this last equality into \eqref{eq:tempCondEnt}, we have that
\begin{align*}
        &\sum_{t = 0}^T \left[ h(X_t | Y^{t}, U^{t-1}) - h(X_t | Y^{t-1}, U^{t-1}) \right.\\
        &\quad\left. + h(X_t | X_{t-1}, Y^{t-1}, U^{t-1}) \right]\\
        &= h(X^{T-1} | Y^{T-1}, U^{T-1}) + h(X_T | Y^T, U^{T-1})\\
        &\quad- h(X^{T-1} | Y^{T-1}, U^{T-1}) + h(X^{T-1} | X_T, Y^T, U^{T-1})\\
        &= h(X_T | Y^T, U^{T-1}) + h(X^{T-1} | X_T, Y^T, U^{T-1})\\
        &= h(X^T | Y^T, U^{T-1})
\end{align*}
where the last line follows due to the relationship between joint and conditional entropies.
Thus, \eqref{eq:stageAdditiveExp} holds for $T$ and the proof of \eqref{eq:stageAdditiveExp} via induction is complete.
\end{proof}

\begin{remark}
The first two terms in the summands in \eqref{eq:stageAdditiveExp} can be written as conditional mutual information $I(X_t;Y_t|Y^{t-1},U^{t-1})$. However, we prefer to keep \eqref{eq:stageAdditive} and \eqref{eq:stageAdditiveExp} in entropy form, to bring out the connection to state estimates computed by Bayesian filters.
Given measurement and control realisations $y^t$ and $u^{t-1}$, Bayesian filtering algorithms such as the Kalman filter or the hidden Markov model (HMM) filter can be employed to efficiently compute conditional state densities (or distributions) $p(x_t | y^t, u^{t-1})$, $p(x_t | y^{t-1}, u^{t-1})$, and $p(x_t, x_{t-1} | y^{t-1}, u^{t-1})$ which may be used to compute the entropies in \eqref{eq:stageAdditive}.
\end{remark}

We next introduce the generic Bayesian filter for the system \eqref{eq:sys} before using it and Lemma \ref{lemma:stageAdditive} to reformulate \eqref{eq:smoothingAverse} as a fully observed stochastic optimal control problem.






\subsection{Information State and Stochastic Optimal Control Formulation}

Let us define $\pi_{t} : \mathcal{X} \mapsto [0,\infty)$ as the \emph{information state}, that is, the conditional state probability density function (or probability mass function when $\mathcal{X}$ is discrete) $\pi_{t}(x_t) \triangleq p(x_t | y^t, u^{t-1})$ for $0 \leq t < T$.
Given measurement and control realisations $y^T$ and $u^{T-1}$ of \eqref{eq:sys}, the information state $\pi_{t}$ satisfies the recursion
\begin{align}
    \label{eq:bayesTemp}
    \pi_{t}(x_t)
    &= \dfrac{p(y_t | x_t, u_{t-1}) \pi_{t | t-1}(x_t)}{\int_{\mathcal{X}} p(y_t | x_t, u_{t-1}) \pi_{t | t-1}(x_t) \, dx_t}
\end{align}
for all $x_t \in \mathcal{X}$ and $0 < t \leq T$ from an initial state probability density $\pi_0$ where the predicted state probability density $\pi_{t | t-1}(x_t) \triangleq p(x_t | y^{t-1}, u^{t-1})$ is obtained by marginalising the joint density $\bar{\pi}_{t | t-1} (x_t, x_{t-1}) \triangleq p(x_t, x_{t-1} | y^{t-1}, u^{t-1})$ over $x_{t-1}$ noting that
\begin{align}
    \label{eq:bayesianPred}
    \bar{\pi}_{t | t-1}(x_t,x_{t-1})
    &= p(x_t | x_{t-1}, u_{t-1}) \pi_{t-1}(x_{t-1}).
\end{align}
As shorthand, we shall let $\Pi$ denote the mapping defined by the Bayesian filter recursion \eqref{eq:bayesTemp}, namely,
\begin{align}
    \label{eq:bayesianFilter}
    \pi_{t}
    &= \Pi(\pi_{t-1}, u_{t-1}, y_t)
\end{align}
for $0 < t \leq T$ from an initial information state $\pi_0$.

Given $\pi_{t}$, $\pi_{t-1}$, and $\bar{\pi}_{t | t-1}$, the entropies in the expectation of \eqref{eq:stageAdditive} may be written as
\begin{align*}
    &h(X_t | X_{t-1}, y^{t-1}, u^{t-1})\\
    &= - \int_{\mathcal{X} \times \mathcal{X}} \bar{\pi}_{t | t-1}(x_t, x_{t-1}) \log \dfrac{\bar{\pi}_{t | t-1}(x_t, x_{t-1})}{\pi_{t-1}(x_{t-1})} \, dx_t \, dx_{t-1},
\end{align*}
together with
\begin{align*}
    h(X_t | y^t, u^{t-1})
    &= - \int_{\mathcal{X}} \pi_t(x_t) \log \pi_t(x_t) \, dx_t,
\end{align*}
and
\begin{align*}
    h(X_t | y^{t-1}, u^{t-1})
    &= - \int_{\mathcal{X}} \pi_{t | t-1}(x_t) \log \pi_{t | t-1}(x_t) \, dx_t
\end{align*}
(where the integrals become summations if the state space $\mathcal{X}$ is discrete).
Importantly, these entropies may equivalently be viewed as functions only of $\pi_{t-1}$, $y_t$ and $u_{t-1}$ since the Bayesian filtering equations \eqref{eq:bayesianPred} and \eqref{eq:bayesianFilter} imply that $\bar{\pi}_{t | t-1}$ and $\pi_{t}$ are functions of $\pi_{t-1}$ given the measurement $y_t$ and control $u_{t-1}$.
We may therefore define the function
\begin{align}\notag
    \tilde{r} \left( \pi_{t-1}, u_{t-1}, y_t\right)
    &\triangleq h(X_t | y^{t}, u^{t-1}) - h(X_t | y^{t-1}, u^{t-1})\\\label{eq:gtilde}
    &\quad+ h(X_t | X_{t-1}, y^{t-1}, u^{t-1})
\end{align}
so that the summation in \eqref{eq:stageAdditive} can be written as
\begin{align}
    \label{eq:stage_additive_gtilde}
    \begin{split}
        &h(X^T | Y^T, U^{T-1})\\
        &\quad= E_{Y^{T}, U^{T-1}} \left[ \sum_{t = 0}^T \tilde{r} \left( \pi_{t-1}, u_{t-1}, y_t\right) \right]
    \end{split}
\end{align}
with $\tilde{r} \left( \pi_{-1}, u_{-1}, y_0\right) \triangleq h(X_0|y_0)$.
Our main result follows.



\begin{theorem}
\label{theorem:ocp}
Let us define 
\begin{align*}
    r(\pi_{t}, u_{t})
    &\triangleq E_{Y_{t+1}, X_{t}} \left[ \tilde{r} \left( \pi_{t}, u_{t}, y_{t+1}\right) - \gamma c_t (x_{t}, u_{t}) | \pi_{t}, u_{t} \right]
\end{align*}
for $0 \leq t < T$.
Then, our smoothing-averse control problem \eqref{eq:smoothingAverse} is equivalent to the stochastic optimal control problem:
\begin{align}
\label{eq:ocp}
\begin{aligned}
&\sup & & E_{\pi^{T-1}, U^{T-1}} \left[ \left. \sum_{t = 0}^{T-1} r \left( \pi_{t}, u_{t} \right) \right| \pi_0 \right]\\ 
&\mathrm{s.t.} & &  \pi_{t+1} = \Pi\left( \pi_{t}, u_t, y_{t+1} \right)\\
& & & Y_{t+1} \sim p(y_{t+1} | \pi_{t}, u_{t})\\
& & & \mathcal{U} \ni U_t \sim \bar{\mu}_t(\pi_t)
\end{aligned}
\end{align}
where the optimisation is over policies $\bar{\mu} \triangleq \{\bar{\mu}_t : 0 \leq t < T\}$ which are functions of $\pi_t$ in the sense that $\bar{\mu}_t(\pi_t) \triangleq p(u_t | \pi_t)$.
\end{theorem}
\begin{proof}
Substituting \eqref{eq:stage_additive_gtilde} into \eqref{eq:smoothingAverse} gives that
\begin{align*}
    &h(X^T | Y^T, U^{T-1}) - \gamma E \left[ \sum_{t = 0}^{T-1} c_t \left(x_t, u_t\right) \right]\\
    &\;= h(X_0 | Y_0) + \sum_{t = 0}^{T-1} E \left[ \tilde{r} \left( \pi_{t}, u_{t}, y_{t+1} \right) - \gamma c_t (x_t, u_t) \right]
\end{align*}
where we have used the fact that $E[\tilde{r} \left( \pi_{-1}, u_{-1}, y_0\right)] = E[h(X_0|y_0)] = h(X_0 | Y_0)$ and the expectations $E[\cdot]$ are over $X^T$, $Y^T$, and $U^{T-1}$ (since $\pi_t$ is a function of $Y^t$ and $U^{t-1}$).
The linearity and tower properties of expectations give that
\begin{align*}
    &E \left[ \tilde{r} \left( \pi_{t}, u_{t}, y_{t+1}\right) - \gamma c_t(x_t, u_t) \right]\\
    &\quad= E_{\pi_t, U_t} \left[ E_{Y_{t+1}} \left[ \tilde{r} \left( \pi_{t}, u_{t}, y_{t+1}\right) | \pi_{t}, u_{t} \right] \right.\\
    &\qquad \left. - \gamma E_{X_t} \left[ c_t(x_t, u_t) | \pi_{t}, u_{t} \right] \right] 
    = E_{\pi_t, U_t} \left[ r(\pi_{t}, u_{t}) \right]
\end{align*}
where the last equality follows by recalling the definition of $r$ and to compute the inner conditional expectations in the second line we note that
\begin{align}
    \label{eq:measurementsState}
    \begin{split}
        p(y_{t+1} | \pi_{t}, u_{t})
        &= \int_{\mathcal{X}\times\mathcal{X}} p(y_{t+1} | x_{t+1}, u_{t}) \\
        &\qquad \times p(x_{t+1} | x_{t}, u_{t}) \pi_{t}(x_{t}) \, dx_{t+1} \, dx_{t}
    \end{split}
\end{align}
and $p(x_{t} | \pi_{t}, u_{t}) = \pi_t(x_t)$.
Thus,
\begin{align*}
    &h(X^T | Y^T, U^{T-1}) - \gamma E \left[ \sum_{t = 0}^{T-1} c_t \left(x_t, u_t\right) \right]\\
    &\quad= h(X_0 | Y_0) + E_{\pi^{T-1}, U^{T-1}} \left[ \sum_{t = 0}^{T-1} r(\pi_t, u_t) \right].
\end{align*}
Finally, $h(X_0 | Y_0)$ is constant, and so it can be omitted from maximisation in our smoothing-averse problem.
Furthermore, $\pi_t$ is a sufficient statistic for the information $(y^t,u^{t-1})$ (see \cite[Section 5.4.1]{Bertsekas2005}) and so it suffices to consider the policies $\bar{\mu}$.
The proof is complete.
\end{proof}

\begin{remark}
\label{remark:theorem1}
The right-hand side of \eqref{eq:gtilde} can be rewritten in the form $h(X_t|X_{t-1},y^{t-1},u^{t-1}) - I(X_t;y_t|y^{t-1},u^{t-1})$. Thus in broad terms,  the optimal policy seeks to increase the unpredictability in the forward one-step-ahead dynamics, while also reducing the information about the current state gained from each successive measurement.
Note also that the stochastic optimal control problem \eqref{eq:ocp} is fully observed in terms of the state estimate $\pi_t$ from the Bayesian filter \eqref{eq:bayesianFilter}, i.e.\ the information state.
It is therefore possible to solve our smoothing-averse problem in the form \eqref{eq:ocp} using standard dynamic programming tools for (perfect information) stochastic optimal control (cf.~\cite{Bertsekas2005}).
\end{remark}

\subsection{Dynamic Programming Equations}
To present dynamic programming equations for solving \eqref{eq:ocp}, we note that it suffices to consider policies $\bar{\mu}$ that are deterministic functions of the conditional state distribution $\pi_t$ in the sense that $u_t = \bar{\mu}_t(\pi_t)$ (see \cite[Proposition 8.5]{Bertsekas1996} and \cite{Bertsekas2005} for details).
The value functions of the optimal control problem \eqref{eq:ocp} are then defined as
\begin{align*}
    J_t(\pi_t)
    \triangleq \sup E_{\pi_{t+1}^{T-1}, U_{t}^{T-1}} \left[ \left. \sum_{\ell = t}^{T-1} r \left( \pi_{\ell}, u_{\ell} \right) \right| \pi_{t} \right]
\end{align*}
for $0 \leq t < T$ with the optimisation subject to the same constraints as \eqref{eq:ocp} and where $\pi_{t+1}^{T-1} \triangleq \{ \pi_{t+1}, \ldots, \pi_{T-1} \}$ and $U_{t}^{T-1} \triangleq \{ U_t, \ldots, U^{T-1} \}$.
The value functions satisfy the dynamic programming recursions
\begin{align}
    \label{eq:deterministicValueFunction}
    \begin{split}
    J_t(\pi_t)
    &= \sup_{u_t \in \mathcal{U}} \left\{ r \left( \pi_{t}, u_{t} \right) \right. \\
    &\qquad \left. + E_{Y_{t+1}} \left[ J_{t+1}(\Pi(\pi_{t}, u_t, y_{t+1})) | \pi_t, u_t \right] \right\}
    \end{split}
\end{align}
for $0 \leq t < T$ where $J_T(\pi_T) \triangleq 0$ and the optimal policy is
\begin{align}
    \label{eq:dpPolicy}
    \begin{split}
    \bar{\mu}_t^*(\pi_t) 
    &= u_t^*
    \in \argsup_{u_t \in \mathcal{U}} \left\{ r \left( \pi_{t}, u_{t} \right) \right. \\
    &\; \qquad \left. + E_{Y_{t+1}} \left[ J_{t+1}(\Pi(\pi_{t}, u_t, y_{t+1})) | \pi_t, u_t \right] \right\}.
    \end{split}
\end{align}

Directly solving the dynamic programming recursions \eqref{eq:deterministicValueFunction} is typically only feasible numerically and only on problems with small state, measurement, and control spaces where the conditional state distributions $\pi_t$ have a finite-dimensional representation (such as when the system \eqref{eq:sys} is linear-Gaussian).
Numerous approximate dynamic programming techniques do exist (see \cite{Krishnamurthy2016, Bertsekas2005} and references therein) with many being concerned with approximating the future value functions $J_{t+1}$.
We shall next consider examples in which we are able to use approximate dynamic programming approaches to solve our smoothing-averse problem \eqref{eq:smoothingAverse}.

\section{Examples and Simulation Results}
\label{sec:examples}
In this section, we illustrate the applicability of our smoothing-averse control problem to privacy in cloud-based control and covert robotic navigation.

\subsection{Privacy in Cloud-based Control}

For our first example, we consider the cloud-based control setting illustrated in Fig.~\ref{fig:cloud_based} and based on the scheme described in \cite{Tanaka2017,Nekouei2019}.
In this cloud-based control setting, the client seeks to have the state process $X_t$ controlled by the cloud service without disclosing the state trajectory $X^T$ directly.
The client provides the cloud service with outputs $Y_t$ of a privacy filter and the cloud service computes and returns the control $U_t = \mu_t(Y^t, U^{t-1})$ using an agreed upon policy $\mu = \{\mu_t : 0 \leq t < T\}$.
The cloud service has knowledge of the policy $\mu$, and the system and privacy filter functions $f_t$ and $g_t$.
The problem the client faces in preserving the privacy of the state trajectory $X^T$ is consistent with our smoothing-averse control problem \eqref{eq:smoothingAverse} with the cloud service being the adversary and having measurements $Z_t = (Y_t, U_{t-1})$.

\begin{figure}
    \centering
    \includegraphics[width=0.9\columnwidth]{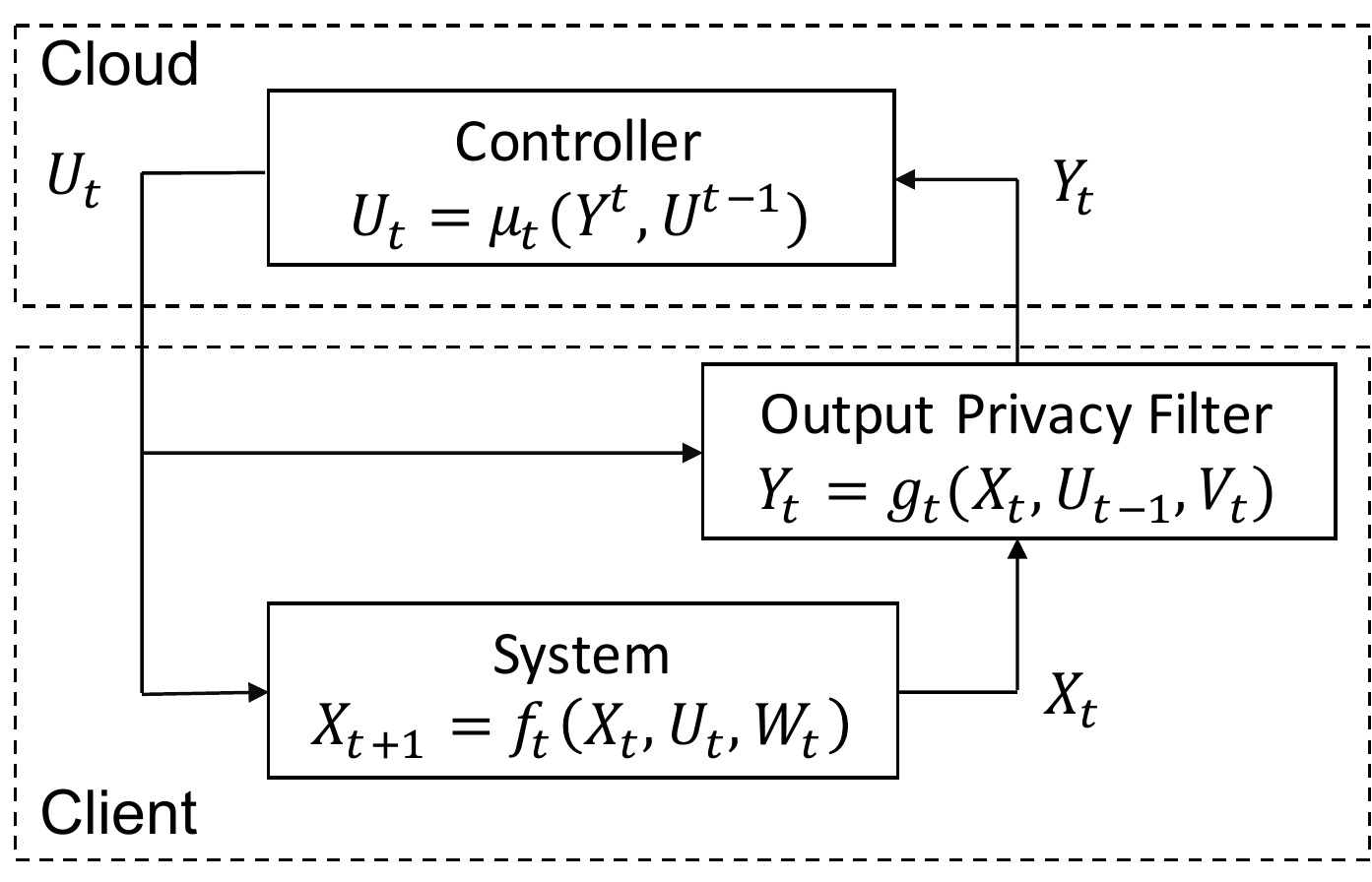}
    \caption{Cloud-based control scheme with output privacy filter based on the scheme described in \cite{Tanaka2017,Nekouei2019}.}
    \label{fig:cloud_based}
\end{figure}

\subsubsection{Simulation Example}
To simulate our smoothing-averse control policy \eqref{eq:dpPolicy} in this cloud-based control setting, we consider a controlled three state Markov chain for the state process $X_t \in \mathcal{X} \triangleq \{1, 2, 3\}$ for $ 0 \leq t \leq T = 10$.
The state dynamics are summarised by three state transition matrices $A({U_t}) \in \mathbb{R}^{3 \times 3}$ that are selected through the choice of controls $U_t \in \{1, 2, 3\}$ and given by
\begin{align*}
    A(1) = 
    \begin{bmatrix}
        0.8 & 0.8 & 0.1\\
        0.1 & 0.1 & 0.8\\
        0.1 & 0.1 & 0.1
    \end{bmatrix}, \;
    A(2) = 
    \begin{bmatrix}
        0.1 & 0.1 & 0.1\\
        0.8 & 0.1 & 0.1\\
        0.1 & 0.8 & 0.8
    \end{bmatrix}
\end{align*}
and
\begin{align*}
    A(3) = 
    \begin{bmatrix}
        0.9 & 0.05 & 0.05\\
        0.05 & 0.9 & 0.05\\
        0.05 & 0.05 & 0.9
    \end{bmatrix}.
\end{align*}
Here, the element in the $i$th row and $j$th column of the state transition matrices correspond to the probabilities $A^{ij}(U_t = u_t) = P(X_{t+1} = i | X_{t} = j, U_t = u_t)$.
The outputs of the privacy filter (or measurements of the states) $Y_t \in \mathcal{Y} \subset \mathbb{R}$ are Gaussian with conditional distribution $p(y_t | X_t = i) = \mathcal{N}(y_t; \nu^i, 1)$ where $\mathcal{N}(\cdot; \nu^i, 1)$ is the univariate Gaussian probability density with mean $\nu^i \in \{1, 3, 5\}$ and unit variance.
Here, $\nu^1 = 1$ for $X_t = 1$, $\nu^2 = 3$ for $X_t = 2$, and $\nu^3 = 5$ for $X_t = 3$.

\subsubsection{Simulation Results}

For the purpose of simulations, we took $c_t(x_t, u_t) = 0$ and $\gamma = 0$ and numerically solved the dynamic programming recursions \eqref{eq:deterministicValueFunction} for the optimal smoothing-averse policy \eqref{eq:dpPolicy} using Markov chain approximation ideas (cf.~\cite[Chapter 4]{Kushner2013}).
Specifically, we gridded the space of conditional state distributions $\pi_t \in \{ \pi = [\pi^1 \; \pi^2 \; \pi^3]': \pi^i \geq 0, \; \sum_{i=1}^3 \pi^i = 1 \}$ with a resolution of $0.00001$ in each component and computed approximate transition dynamics for each control $U_t$ between these discretised state distributions given the measurement process \eqref{eq:measurementsState} (here, $'$ denotes matrix transposition).
We then solved \eqref{eq:dpPolicy} to find and store the approximate optimal policy at all points in the grid.
Whilst the process of constructing the Markov chain approximation and the approximate optimal policy is computationally expensive, it is only performed once offline.
During operation, the optimal control is determined using the approximate policy via a simple look-up by:
\begin{enumerate}
    \item 
    Computing the conditional state distribution $\pi_t$ using a HMM filter (cf.~\cite{Elliott1995}); and,
    \item
    Projecting $\pi_t$ to the closest point in the grid of the Markov chain approximation and looking-up the control at this point in the precomputed approximate policy.
\end{enumerate}

For comparison, we also (approximately) solved the dynamic programming recursions \eqref{eq:deterministicValueFunction} with
\begin{align*}
    \tilde{r}(\pi_{t}, u_{t}, y_{t+1})
    &= h(X_{t+1} | y^{t+1}, u^t) - h(X_{t+1} | y^t, u^t)
\end{align*}
in $r(\pi_t, u_t)$ instead of \eqref{eq:gtilde} in order to minimise the information gain provided directly by the measurements (cf.\ Remark \ref{remark:theorem1}).
This \emph{Min.\ Information Gain} approach is similar to the minimum directed information approach proposed in \cite{Tanaka2017} and establishing their precise relationship remains a topic of our ongoing research.

We performed 200 Monte Carlo simulations of both our smoothing-averse control policy and that minimising the information gain due to measurements (as in \cite{Tanaka2017}).
The initial state of the system in each simulation was selected from a uniform distribution over $\mathcal{X}$ and $T = 10$.
Table \ref{tbl:illustrativeExample} summarises the estimated smoother entropies under both approaches as well as the estimated average probability of error of maximum a posteriori (MAP) estimates of the states $\{X_1, \ldots, X_{10} \}$.
The smoother entropies $h(X^T | Y^{T}, U^{T-1})$ were specifically estimated by averaging the entropies $h(X^T | y^{T}, u^{T-1})$ of the posterior state distributions $p(x^{T} | y^{T}, u^{T-1})$ over the 200 Monte Carlo runs, whilst the average MAP error probabilities were estimated by counting the number of times the state with the greatest probability in the (marginal) posterior $p(x_t | y^{T}, u^{T-1})$ differed from the true state (and averaging over $T = 10$ and the 200 Monte Carlo runs).

The results reported in Table \ref{tbl:illustrativeExample} suggest that use of our smoothing-averse control policy results in a greater smoother entropy $h(X^T | Y^{T}, U^{T-1})$ (i.e., the average entropy of the posterior state distributions $p(x^{T} | y^{T}, u^{T-1})$) than use of the policy that minimises only the measurement information gain.
In these simulations, we also see that the average MAP error probability under our smoothing-averse policy is $3.7\%$ greater than under the minimum information gain policy.
Our smoothing-averse control policy therefore increases the difficulty of inferring the states compared to the minimum information gain policy.
The key reason for the greater smoother entropy and error probabilities under our smoothing-averse policy is that it seeks to increase the unpredictability in the forward one-step-ahead dynamics, while also reducing the information about the current state gained from each successive measurement (as discussed in Remark \ref{remark:theorem1}).
We expect the differences in performance of our smoothing-averse control policy and the minimum information gain policy to persist in cases where both policies have other objectives (i.e., when $c_t(x_t, u_t) \neq 0$) but we leave this to future work.


\begin{table}[t!]
\begin{center}
\caption{Estimated entropy $h(X^T | Y^T, U^{T-1})$ and average maximum a posteriori (MAP) error probabilities from 200 simulations.}
\label{tbl:illustrativeExample}
\begin{tabular}{@{}lcc@{}}
\toprule
\multicolumn{1}{c}{\textbf{Approach}} & \textbf{\begin{tabular}[c]{@{}c@{}}Smoother \\ Entropy\end{tabular}} & \textbf{\begin{tabular}[c]{@{}c@{}}Average MAP\\ Error Prob.\end{tabular}} \\ \midrule
\textbf{Proposed Smoothing-Averse}               & 3.4017                                                                                   & 0.1750                                                                           \\                                              
\textbf{Min. Information Gain} &  2.5505 & 0.1380                                                                  \\\bottomrule
\end{tabular}
\end{center}
\end{table}

\subsection{Covert Robotic Navigation}

For our second example, let us consider a robot navigating to a goal location by localising itself using range and bearing measurements to a number of landmarks with known locations.
An adversary is tracking the robot by also obtaining range and/or bearing measurements from each of the landmarks to the robot and combining these to estimate the robot's position and orientation (e.g., the adversary's sensors may be colocated with the landmarks).
The robot seeks to covertly navigate to its desired goal location whilst making it difficult for the adversary to estimate its trajectory.
This covert navigation problem is consistent with our smoothing-averse control problem \eqref{eq:smoothingAverse} under Assumption \ref{assumption:measurements} when the robot assumes that the adversary possesses the worst-case capability to infer its controls $U_t$, and has access to similar measurements $Y_t$ such that $Z_t = (Y_t, U_{t-1})$.
 
\subsubsection{Simulation Example}
For the purpose of simulations, we consider 
a robot in two dimensions with state vector $X_t = [X_t^1 \; X_t^2 \; X_t^3]'$ where $X_t^1$ and $X_t^2$ are the $x-$ and $y-$ axis positions of the robot in meters and $X_t^3 \in [-\pi, \pi)$ is the robot heading in radians.
The robot controls are $U_t = [U_t^1 \; U_t^2]'$ where $U_t^1 \geq 0$ is its linear speed in meters per second and $U_t^2$ is its angular turn rate in radians per second.
The robot's motion model is
\begin{align*}
    X_{t+1}
    &= X_t + 
    \begin{bmatrix}
        U_t^1 \Delta t \sin X_t^3 \\
        U_t^1 \Delta t \cos X_t^3 \\
        \Delta t U_t^2
    \end{bmatrix}
    + W_t
\end{align*}
where $\Delta t = 1$ is the sampling time.
We consider the map with 5 landmarks shown in Fig.~\ref{fig:covert} and assume that all landmarks are observed by the robot at each time $t$ so that the total measurement vector is given by $Y_t = [Y_t^{1\prime} \; \cdots \; Y_t^{5\prime}]'$.
Here, $Y_t^j$ is the range and bearing measurements to the $j$th landmark satisfying
\begin{align*}
    Y_t^j
    &=
    \begin{bmatrix}
        \sqrt{(m_x^j- X_t^1)^2 + (m_y^j - X_t^2)^2}\\
        \arctan2(m_y^j - X_t^2, m_x^j - X_t^1) - X_t^3
    \end{bmatrix}
    + V_t^j
\end{align*}
with $(m_x^j, m_y^j)$ is the location of the $j$th landmark.
We note that both $W_t$ and $V_t^j$ are zero-mean independent and identically distributed Gaussian process with covariance matrices
\begin{align*}
    \begin{bmatrix}
    0.1^2 & 0 & 0\\
    0 & 0.1^2 & 0\\
    0 & 0 & (\frac{\pi}{180})^2
    \end{bmatrix}
    \text{ and }
    \begin{bmatrix}
        50^2 & 0\\
        0 & (\frac{\pi}{18})^2
    \end{bmatrix}, \text{ respectively.}
\end{align*}

\subsubsection{Simulation Results}

For simulations, we took the goal pose of the robot to be $x_g = [-150 \; 50\; 0]'$ and incorporated it via the cost $c_t(x_t, u_t) = \|x_t - x_g\|^2$ with $\gamma > 0$.
The continuous state and measurement spaces in this example make the direct solution of the dynamic programming equations \eqref{eq:deterministicValueFunction} intractable.
We therefore implemented a suboptimal approach based on receding-horizon ideas and involving both the approximation of the future value in the dynamic programming recursions \eqref{eq:deterministicValueFunction} using a small number of Monte Carlo simulations and a simplified set of controls.
Specifically, at each time $t$, the robot's linear speed control is set to $U_t^1 = 1$ and the robot performs a turn with (potentially zero) rate $U_{t}^2 \in [-\pi, \pi)$ solving
		\begin{align*}
		U_{t}^2
		&= \argmax_{u \in [-\pi, \pi)} E_{Y_{t+1}} \left[ J_{t+1}(\Pi(\pi_t, [1, u]', y_{t+1})) | \pi_t \right]
	\end{align*}
	which approximates the optimal policy \eqref{eq:dpPolicy} by:
	\begin{enumerate}
		\item 
		Treating the instantaneous cost $r \left( \pi_{t}, u_{t} \right)$ as constant (and hence omitting it) since moving during a single time step does not significantly change the robot's position; and,
		\item
		Computing an approximation of the expected future value of $J_{t+1}$ by simulating the robot forward in time over a fixed horizon of $10$ time steps and averaging the value of the objective function over 10 different simulations, namely,
\begin{align*}
    &E_{Y_{t+1}} \left[ J_{t+1}(\Pi(\pi_t, [1, u]', y_{t+1})) | \pi_t \right]\\
    &\quad\approx \dfrac{1}{10} \sum_{n = 1}^{10} \sum_{k = t+1}^{t + 10} \left[ \tilde{r}(\pi_{k}^{(n)}, [1, u]', y_{k+1}^{(n)}) \right.\\
    &\qquad\qquad\qquad\qquad\qquad\left. - \gamma \| x_{k}^{(n)} - x_g\|^2 \right]
\end{align*}
where the superscripts denote quantities from the $n$th simulation run.
The entropies in $\tilde{r}$ are computed using the covariance matrices from an implementation of the recursive landmark-based extended Kalman filter localisation algorithm detailed in \cite[Chapter 7]{Thrun2005} together with the equations for the entropy of multivariate Gaussian distributions (see \cite[Theorem 8.4.1]{Cover2006}).
\end{enumerate}


We performed 25 independent simulations of our approximate smoothing-averse control problem with $\gamma = 100$ and $\gamma = 0.06$ to examine the impact of prioritising getting closer to the goal pose (with $\gamma = 100$) versus increasing the smoother entropy $h(X^T | Y^T, U^{T-1})$ (with $\gamma = 0.06$).
The mean robot path (and the standard deviation of the robot's path) over these simulations is shown in Fig.~\ref{fig:direct} for $\gamma = 100$ and Fig.~\ref{fig:smoothingaverse} for $\gamma = 0.06$.
From Fig.~\ref{fig:direct}, we see that when not considering the smoother entropy, the robot will proceed directly to the goal position and pass close to the landmarks.
In contrast, Fig.~\ref{fig:smoothingaverse} suggests that the robot is more varied in its path and often seeks to avoid passing near landmarks when attempting to increase the smoother entropy since the farther the robot is from a landmark, the more ambiguous the robot's position given a particular range and bearing measurement.

\begin{figure}[t!]
     \centering
     \begin{subfigure}[t!]{\columnwidth}
         \centering
         \includegraphics[width = 0.9\columnwidth]{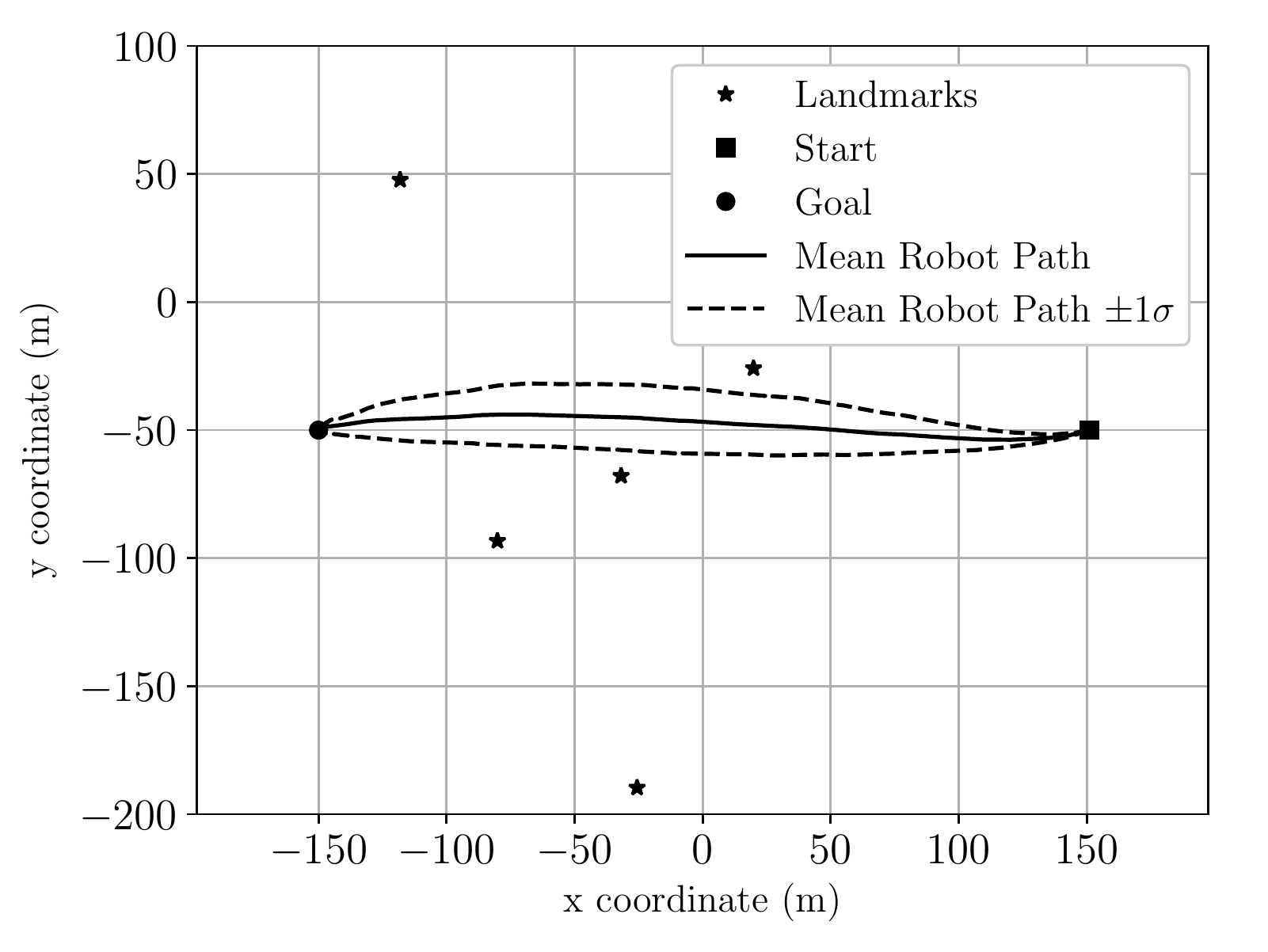}
         \caption{}
         \label{fig:direct}
     \end{subfigure}
     \vfill
     \begin{subfigure}[t!]{\columnwidth}
         \centering
         \includegraphics[width = 0.9\columnwidth]{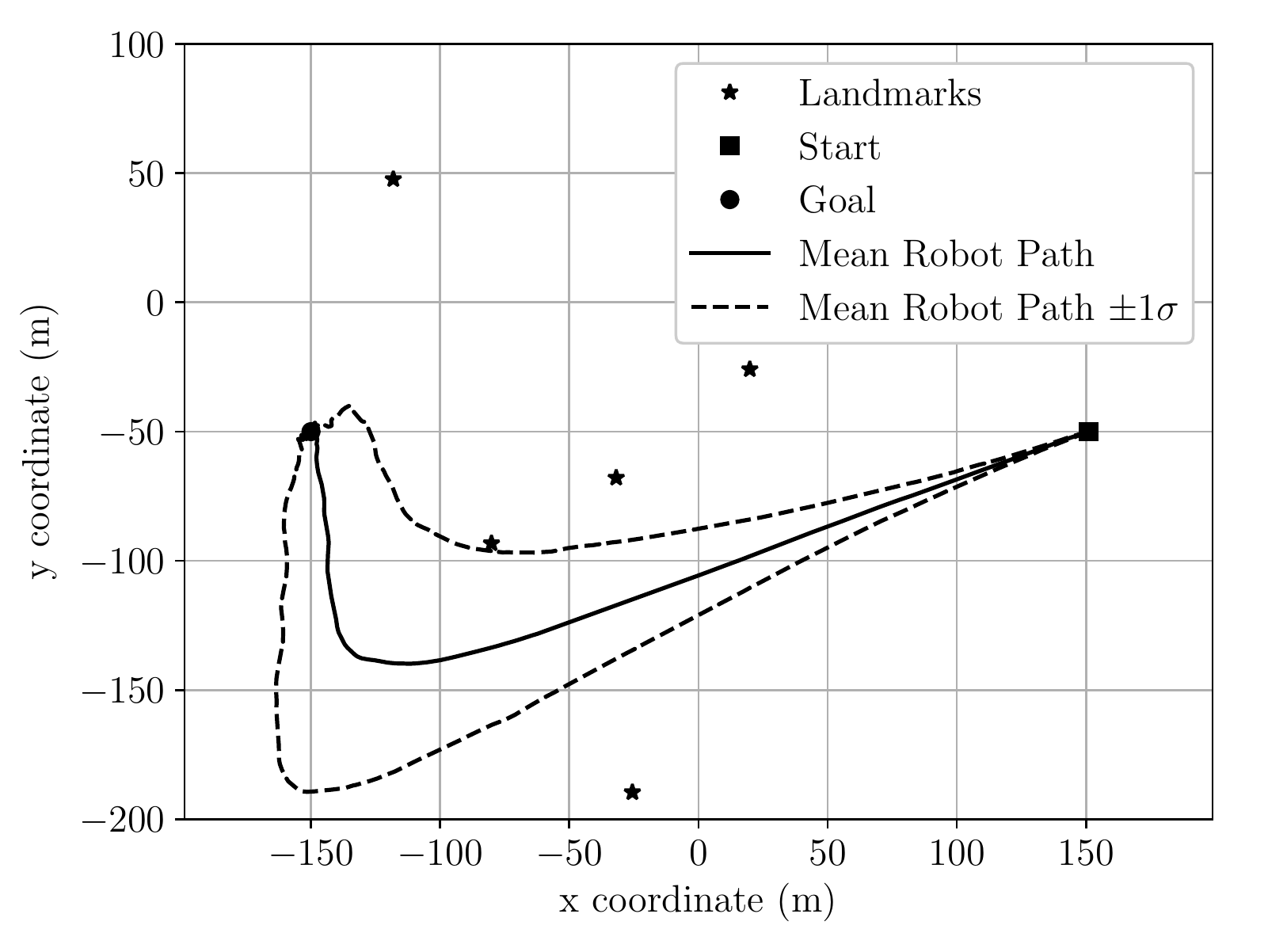}
         \caption{}
         \label{fig:smoothingaverse}
     \end{subfigure}
        \caption{Covert robotic navigation example with (a) $\gamma = 100$ to heavily penalise distance from the goal location, and (b) $\gamma = 0.06$ to jointly penalise distance from goal and increase smoother entropy $h(X^T | Y^T, U^{T-1})$.}
        \label{fig:covert}
\end{figure}

\section{Conclusion and Future Work}
\label{sec:conclusion}

We posed a new smoothing-averse control problem involving the maximisation of the entropy of the state trajectory estimate from a (fixed-interval) Bayesian smoother.
We showed that the smoother entropy for general nonlinear state-space models has an additive form that enables smoothing-averse control to be reformulated as an optimal control problem solvable via dynamic programming.
We illustrated the applicability of smoothing-averse control to privacy in cloud-based control and covert robotic navigation.

Future work will investigate the efficient (numerical) solution of the optimal control problem and dynamic programming equations that have arisen in our consideration of smoothing-averse control.
We also plan to investigate provable guarantees on the privacy and covertness provided by smoothing-averse control schemes, and best-response (i.e., game-theoretic) strategies for adversaries to counter smoothing-averse schemes.
In particular, the smoothing-averse scheme could be extended to pose novel adversarial inverse smoothing problems in a similar vain to recent inverse filtering problems that involve estimating (or countering the estimation) of state posterior distributions from observed actions or controls (cf.\ \cite{Lourenco2020,Krishnamurthy2019,Mattila2020}).
Similarly, as in recent work \cite{Lourenco2020}, it would be interesting to investigate optimising alternative privacy functions to the smoother entropy to conceal entire state trajectories.


\section*{Acknowledgement}

We thank Iman Shames for his feedback and discussion on a draft of this paper.

\bibliographystyle{IEEEtran}
\bibliography{IEEEabrv,Library}

\begin{thebibliography}{10}
\providecommand{\url}[1]{#1}
\csname url@samestyle\endcsname
\providecommand{\newblock}{\relax}
\providecommand{\bibinfo}[2]{#2}
\providecommand{\BIBentrySTDinterwordspacing}{\spaceskip=0pt\relax}
\providecommand{\BIBentryALTinterwordstretchfactor}{4}
\providecommand{\BIBentryALTinterwordspacing}{\spaceskip=\fontdimen2\font plus
\BIBentryALTinterwordstretchfactor\fontdimen3\font minus
  \fontdimen4\font\relax}
\providecommand{\BIBforeignlanguage}[2]{{%
\expandafter\ifx\csname l@#1\endcsname\relax
\typeout{** WARNING: IEEEtran.bst: No hyphenation pattern has been}%
\typeout{** loaded for the language `#1'. Using the pattern for}%
\typeout{** the default language instead.}%
\else
\language=\csname l@#1\endcsname
\fi
#2}}
\providecommand{\BIBdecl}{\relax}
\BIBdecl

\bibitem{Karabag2019}
M.~O. {Karabag}, M.~{Ornik}, and U.~{Topcu}, ``Least inferable policies for
  {Markov} decision processes,'' in \emph{American Control Conference (ACC)},
  2019, pp. 1224--1231.

\bibitem{Li2019}
N.~Li, I.~Kolmanovsky, and A.~Girard, ``Detection-averse optimal and
  receding-horizon control for {Markov} decision processes,'' \emph{arXiv
  preprint arXiv:1908.07691}, 2019.

\bibitem{Savas2020}
Y.~{Savas}, M.~{Ornik}, M.~{Cubuktepe}, M.~O. {Karabag}, and U.~{Topcu},
  ``Entropy maximization for {Markov} decision processes under temporal logic
  constraints,'' \emph{IEEE Trans. on Automatic Control}, vol.~65, no.~4, pp.
  1552--1567, 2020.

\bibitem{Hibbard2019}
M.~{Hibbard}, Y.~{Savas}, B.~{Wu}, T.~{Tanaka}, and U.~{Topcu}, ``Unpredictable
  planning under partial observability,'' in \emph{58th IEEE Conference on
  Decision and Control (CDC)}, 2019, pp. 2271--2277.

\bibitem{Krishnamurthy2019}
V.~{Krishnamurthy} and M.~{Rangaswamy}, ``{How to Calibrate Your Adversary's
  Capabilities? Inverse Filtering for Counter-Autonomous Systems},'' \emph{IEEE
  Transactions on Signal Processing}, vol.~67, no.~24, pp. 6511--6525, 2019.

\bibitem{Mattila2020}
R.~{Mattila}, C.~R. {Rojas}, V.~{Krishnamurthy}, and B.~{Wahlberg}, ``{Inverse
  Filtering for Hidden Markov Models With Applications to Counter-Adversarial
  Autonomous Systems},'' \emph{IEEE Transactions on Signal Processing},
  vol.~68, pp. 4987--5002, 2020.

\bibitem{Lourenco2020}
I.~{Lourenço}, R.~{Mattila}, C.~R. {Rojas}, and B.~{Wahlberg}, ``How to
  protect your privacy? a framework for counter-adversarial decision making,''
  in \emph{59th IEEE Conference on Decision and Control (CDC)}, 2020, pp.
  1785--1791.

\bibitem{Mochaourab2018}
R.~{Mochaourab} and T.~J. {Oechtering}, ``Private filtering for hidden {Markov}
  models,'' \emph{IEEE Signal Processing Letters}, vol.~25, no.~6, pp.
  888--892, 2018.

\bibitem{Tanaka2017}
T.~{Tanaka}, M.~{Skoglund}, H.~{Sandberg}, and K.~H. {Johansson}, ``Directed
  information and privacy loss in cloud-based control,'' in \emph{American
  Control Conference (ACC)}, 2017, pp. 1666--1672.

\bibitem{Li2018}
S.~{Li}, A.~{Khisti}, and A.~{Mahajan}, ``Information-theoretic privacy for
  smart metering systems with a rechargeable battery,'' \emph{IEEE Trans. on
  Information Theory}, vol.~64, no.~5, pp. 3679--3695, 2018.

\bibitem{Poor2017}
H.~V. Poor, \emph{Privacy in the Smart Grid: Information, Control, and
  Games}.\hskip 1em plus 0.5em minus 0.4em\relax Cambridge University Press,
  2017, p. 499–518.

\bibitem{Farokhi2018}
F.~{Farokhi} and H.~{Sandberg}, ``Fisher information as a measure of privacy:
  Preserving privacy of households with smart meters using batteries,''
  \emph{IEEE Trans. on Smart Grid}, vol.~9, no.~5, pp. 4726--4734, 2018.

\bibitem{Leong2019}
A.~S. {Leong}, D.~E. {Quevedo}, D.~{Dolz}, and S.~{Dey}, ``Information bounds
  for state estimation in the presence of an eavesdropper,'' \emph{IEEE Control
  Systems Letters}, vol.~3, no.~3, pp. 547--552, 2019.

\bibitem{Tanaka2018}
T.~{Tanaka}, P.~M. {Esfahani}, and S.~K. {Mitter}, ``{LQG} control with minimum
  directed information: Semidefinite programming approach,'' \emph{IEEE Trans.
  on Automatic Control}, vol.~63, no.~1, pp. 37--52, 2018.

\bibitem{Nekouei2019}
E.~Nekouei, T.~Tanaka, M.~Skoglund, and K.~H. Johansson,
  ``Information-theoretic approaches to privacy in estimation and control,''
  \emph{Annual Reviews in Control}, 2019.

\bibitem{Farokhi2020}
F.~Farokhi, Ed., \emph{Privacy in Dynamical Systems}.\hskip 1em plus 0.5em
  minus 0.4em\relax Springer, 2020.

\bibitem{Sandberg2015}
H.~{Sandberg}, G.~{Dán}, and R.~{Thobaben}, ``Differentially private state
  estimation in distribution networks with smart meters,'' in \emph{54th IEEE
  Conference on Decision and Control (CDC)}, 2015, pp. 4492--4498.

\bibitem{Hale2015}
M.~Hale and M.~Egerstedty, ``Differentially private cloud-based multi-agent
  optimization with constraints,'' in \emph{American Control Conference
  (ACC)}.\hskip 1em plus 0.5em minus 0.4em\relax IEEE, 2015, pp. 1235--1240.

\bibitem{Murguia2021}
C.~{Murguia}, I.~{Shames}, F.~{Farokhi}, D.~{Nešić}, and H.~V. {Poor}, ``On
  privacy of dynamical systems: An optimal probabilistic mapping approach,''
  \emph{IEEE Transactions on Information Forensics and Security}, vol.~16, pp.
  2608--2620, 2021.

\bibitem{Briers2010}
M.~Briers, A.~Doucet, and S.~Maskell, ``Smoothing algorithms for state--space
  models,'' \emph{Annals of the Institute of Statistical Mathematics}, vol.~62,
  no.~1, p.~61, 2010.

\bibitem{Marzouqi2006}
M.~S. Marzouqi and R.~A. Jarvis, ``New visibility-based path-planning approach
  for covert robotic navigation,'' \emph{Robotica}, vol.~24, no.~6, pp.
  759--773, 2006.

\bibitem{Marzouqi2011}
------, ``Robotic covert path planning: A survey,'' in \emph{2011 IEEE 5th
  international conference on robotics, automation and mechatronics
  (RAM)}.\hskip 1em plus 0.5em minus 0.4em\relax IEEE, 2011, pp. 77--82.

\bibitem{Rost2019}
P.~{Boström-Rost}, D.~{Axehill}, and G.~{Hendeby}, ``Informative path planning
  in the presence of adversarial observers,'' in \emph{22th International
  Conference on Information Fusion (FUSION)}, 2019, pp. 1--7.

\bibitem{Savas2018}
Y.~Savas, M.~Ornik, M.~Cubuktepe, and U.~Topcu, ``Entropy maximization for
  constrained {Markov} decision processes,'' in \emph{2018 56th Annual Allerton
  Conference on Communication, Control, and Computing (Allerton)}.\hskip 1em
  plus 0.5em minus 0.4em\relax IEEE, 2018, pp. 911--918.

\bibitem{Bar-Shalom2001}
Y.~Bar-Shalom, X.~Rong~Li, and T.~Kirubarajan, \emph{{Estimation with
  applications to tracking and navigation}}.\hskip 1em plus 0.5em minus
  0.4em\relax New York, NY: John Wiley \& Sons, 2001.

\bibitem{Thrun2005}
S.~Thrun, W.~Burgard, and D.~Fox, \emph{Probabilistic Robotics}.\hskip 1em plus
  0.5em minus 0.4em\relax MIT Press, 2005.

\bibitem{Bertsekas2005}
D.~P. Bertsekas, \emph{Dynamic programming and optimal control},
  {Third}~ed.\hskip 1em plus 0.5em minus 0.4em\relax Belmont, MA: Athena
  Scientific, 1995, vol.~1.

\bibitem{Bertsekas1996}
D.~P. Bertsekas and S.~E. Shreve, \emph{Stochastic Optimal Control: The
  Discrete Time Case}.\hskip 1em plus 0.5em minus 0.4em\relax Athena
  Scientific, 1996.

\bibitem{Krishnamurthy2016}
V.~Krishnamurthy, \emph{Partially observed Markov decision processes}.\hskip
  1em plus 0.5em minus 0.4em\relax Cambridge University Press, 2016.

\bibitem{Kushner2013}
H.~Kushner and P.~G. Dupuis, \emph{Numerical methods for stochastic control
  problems in continuous time}.\hskip 1em plus 0.5em minus 0.4em\relax Springer
  Science \& Business Media, 2013, vol.~24.

\bibitem{Elliott1995}
R.~Elliott, L.~Aggoun, and J.~Moore, \emph{{Hidden Markov Models: Estimation
  and Control}}.\hskip 1em plus 0.5em minus 0.4em\relax New York, NY: Springer,
  1995.

\bibitem{Cover2006}
T.~Cover and J.~Thomas, \emph{{Elements of information theory}}, 2nd~ed.\hskip
  1em plus 0.5em minus 0.4em\relax New York: Wiley, 2006.

\end{thebibliography}

\end{document}